\documentclass[11pt]{article}
\usepackage{graphicx,color,amsmath,amsthm,amssymb}

\newtheorem{theorem}{Theorem}
\theoremstyle{definition}
\newtheorem{definition}{Definition}

\usepackage{theapa, rawfonts}

\begin{document}

\begin{center}
{A Generalized Extensive-Form Fictitious Play Algorithm}
\vskip .2in

{ Tim P. Schulze \\ 
   Department of Mathematics\\
  University of Tennessee\\
  1403 Circle Dr.\\
  Knoxville, TN 37916}
\end{center}

\begin{abstract}
We introduce a simple extensive-form algorithm for finding equilibria
of two-player, zero-sum games. The algorithm is realization equivalent
to a generalized form of Fictitious Play. We compare its performance
to that of a similar extensive-form fictitious play algorithm and a
counter-factual regret minimization algorithm.  All three algorithms
share the same advantages over normal-form fictitious play in terms of
reducing storage requirements and computational complexity.  The new
algorithm is intuitive and  straightforward to implement, making
it an appealing option for those looking for a quick and easy game
solving tool.

\end{abstract}

\section{Introduction}
\label{Introduction}

In recent years there has been a great deal of progress in
computational methods for solving large games. Interest in the subject
stems from both practical applications where AIs, such as self-driving
vehicles, interact with each other and humans, and from a handful
recreational games, such as chess, poker and Go, that are seen as
challenging surrogates for real-world applications, while
simultaneously appealing to a large population of devoted
enthusiasts. In particular, work on the popular variant of poker known
as Texas Hold'em has seen many years of progress culminate in a number
of high-profile success stories. Poker and other card games are
especially challenging, as they are games with imperfect information
and a large number of game states.  The development of the
Counter-Factual Regret Minimization (CFR) algorithm by Zinkevich,
et. al. \cite{ZJBP} marked a significant advance in solving large
extensive-form games, eventually leading to the numerical solution of
the two-player, limit version of Texas Hold'em \cite{BBJT}. This was
followed by other successful AIs that defeated top professional poker
players in heads-up no-limit \cite{BS} and multi-player no-limit
\cite{BS2} Texas Hold'em.

It has been recognized from the earliest days of game theory that
using {\em behavior strategies} is often preferable to {\em mixed
  strategies} for analyzing large games.  Despite this, many
computational methods use mixtures, $\sigma^i(s)$, of pure strategies, $s\in S^i$, that
specify a specific  action to be taken by player $i$ at every game state
that player may encounter:
$\sum_{s \in S^i} \sigma^i(s)=1.
$ Fictitious Play (FP), for example, is one of the oldest computational
methods for solving games \cite{B,R}. In its original formulation, it
is a method for finding a Nash Equilibrium (NE) in two-player, zero-sum,
normal form games.  In this method, the average of the prior play is iteratively updated to 
\begin{equation}\label{FP}
\sigma_{n+1}^i=\left(1-\frac{1}{n+1}\right) \sigma^i_n + \frac{1}{n+1} \beta^i(\sigma_n^{-i}),
\end{equation}
where $-i$ indicates player $i$'s opponent, $\beta^i(\sigma^{-i}_n)$ is a best
response to the opponent's play on the previous time-step:
\begin{equation}
\beta^i(\sigma^{-i}_n) \in \arg\max_{\sigma^i} u^i(\sigma^i,\sigma^{-i}_n),
\end{equation}
and $u^i(\sigma^i,\sigma^{-i})$ is expected utility, accounting for the mixed strategies of both players and the role of a chance.

While the normal form of a game is often preferable for analyzing games in
general, it balloons the computational cost and amount of storage
required for games.  In contrast, behavior strategies provide a more
compact way of representing a strategy by assigning a distribution,
$b(I,a)\ge 0$, over the actions $a \in A(I)$ available at each
information set $I\in {\cal I}$:
$$
\sum_{a\in A(I)} b(I,a) =1,
$$
including those controlled by a chance player, who plays a fixed strategy.
When players have perfect recall---they do not forget any
information they knew in the past---there exists a strategy of
either type equivalent to a given strategy of the other type \cite{K,MSZ}.

The use of behavior strategies was one of the features that allowed
the CFR algorithm and its derivatives, like CFR+ \cite{T}, to achieve
the successes mentioned above.  A key advantage of the behavior
strategy description lies in being able to consider the expected
utility, $U(I,a;b)$, of a specific action, $a\in A(I)$, at a given
information set $I$, where the utility is measured from the perspective of the player who controls $I$, and  $b$ represents the
collection of behavior strategies for all players at all information
sets.  In what follows, we refer to $U(I,a;b)$ as the {\em
  action-utility}, and suppress $b$ when the strategies are clear from
context, writing $U(I,a) \equiv U(I,a;b)$. Similarly, we will use
$U(I,{\bf b}) \equiv U(I,{\bf b};b)$ to indicate the utility of
playing a specific mixture of actions ${\bf b}$ at $I$ that may or may
not be consistent with that dictated by $b$.

More recently Heinrich. et. al. \cite{HLS} develop a version of FP,
which they call Extensive-Form Fictitious Play (XFP), that also takes
advantage of the efficiency of behavior strategies.  In this paper we
introduce a similar algorithm that we refer to as Generalized
Extensive-Form Fictitious Play (GXFP).  In close analogy to FP, GXFP
consists of a sequence of behavior strategies:
\begin{equation}\label{GXFP}
  b_{n+1}=\left(1-\frac{1}{n+1}\right) b_n + \frac{1}{n+1} d(b_n),
\end{equation}
where $d(b_n)$ is the collection of {\em best decisions} that are locally optimized with respect to both the opponent's current strategy
 and a player's own current strategy following actions $a \in A(I)$:
\begin{equation} \label{BD}
  d(I,*;b_n) \in \arg \max_{\bf b} U(I,{\bf b};b_n).
\end{equation}  
We will see that GXFP enjoys the same advantages as CFR and XFP in
terms of how computational cost and storage scale with the size of the
game, but with a simpler and more intuitive implementation. In
practice (\ref{BD}) can be be computed by simply selecting the best
action, and this mimics the way humans think. In particular, expected
value computations for what we have called action utilities are
routinely discussed in the recreational poker literature, but to the
extent humans can really make these calculations they focus on their immediate decision using
their own current strategy and their beliefs about how their opponents
play.

In the next section, we briefly review the CFR and XFP algorithms, and
further introduce GXFP.  In section 3, we show that GXFP is equivalent to a
generalized FP, and therefore inherits its convergence properties. In
section 4, we discuss a benchmark game that generalizes  a
classic model of poker put forward by von Neumann and Morgenstern (vN\&M)
\cite{NM}. In section 5, we use this benchmark game to compare the
performance of the three algorithms.  We summarize and conclude in the
final section.

\section{Algorithms}

We start with a description of elements common to all three algorithms
considered in this paper, and follow this with a discussion of each
algorithm separately.

In games with imperfect information, a player may not know which node
he/she is at, and must analyze their decisions based on the
probability that their opponent's prior play has brought them to a
particular information set. In comparing the expected value of
actions, a player need not consider the probability that their own
prior actions will bring them to that information set. Thus the play
of a player in any such calculation is assumed to have been
consistent with the need to make the decision. For this reason,
$U(I,a)$ is referred to as {\em counter-factual utility} in
\cite{ZJBP}.  This ``play-to-reach'' assumption is common to all of
the algorithms we consider in this paper.

The action-utility can be computed using a basic utility function,  ${\cal U}^i(\ell)$, defined on the set of leaves in the game tree, the  conditional probability of $i$'s opponent's, including the chance player, playing so as to reach a node $x \in I$ controlled by $i$, and the conditional probability of all players playing so as to reach leaf $\ell$ starting from $I$ with action $a$:
\begin{eqnarray}
  U(I,a)  &=&\sum_{x \in I} P^{-i}(x|I) \sum_{\ell \in L_{x,a}} P(\ell|x,a) {\cal U}^i(\ell), \\
&=&\sum_{x \in I} \frac{P^{-i}(x)}{P^{-i}(I)} \sum_{\ell \in L_{x,a}} P(\ell|x,a) {\cal U}^i(\ell), \quad P^{-i}(I)>0,
\end{eqnarray}
where the second sum on each line is over 
leaves, $L_{x,a}$, that can be reached using action $a$ at node
$x$.

The various ``reach'' probabilities can be computed from the behavior strategies $b(I,a)$ and the unique sequence of actions starting at the root
 of the game tree and terminating at a node $x$: 
 $a^x_1, a^x_2,\dots,a^x_{J_x}$, where there are $J_x$ actions along the path leading to $x$: 
\begin{eqnarray}
  P^{-i}(x)&=& \prod_{j=1,I(a_j^x)\notin {\cal I}^i}^{J_x} b(I(a_j^x),a_j^x),\\
  P^{-i}(I) &=&  \sum_{x \in I} \prod_{j=1,I(a_j^x)\notin {\cal I}^i
}^{J_x} b(I(a_j^x),a_j^x),\\
  P(\ell|x,a) &=& \prod_{j=J_x+2}^{J_\ell} b(I(a_j^{\ell}),a_j^{\ell}).
\end{eqnarray} 
Note that the behavior coefficient for $a$ is omitted in the last product, as the probability  is conditioned on that choice.

In the rest of this section, we describe
 the three algorithms considered in this paper.

\subsection{Counter-factual Regret Minimization}

The basic version of CFR put forward in \cite{ZJBP} is now sometimes
referred to as ``vanilla'' CFR, and is a popular entry point for those
 getting started with reinforcement learning (RL).  While the authors
go beyond this version, adapting it to specific features of Texas
Hold'em, and there have been subsequent developments, most notably CFR+
\cite{T}, we will be considering only this basic version.

CFR is based on the notion of {\em regret} for having played
the game according to the current strategy $b(I,*)$ rather than 
taking a specific action $a$ at information set $I$:
$$
U(I,a) -U(I,b(I,*)),
$$
where
$$
U(I,b(I,*))=\sum_{a \in A(I)} b(I,a) U(I,a).
$$
More specifically, CFR maintains the average regret, weighted
by the opponent's reach probabilities $P^{-i}(I;b_n)$:
\begin{equation}
R_n(I,a)=\frac{1}{n} \sum_{k=1}^n P^{-i}(I;b_k)
\left(U(I,a;b_k) -U(I,b_k(I,*);b_k)
\right).
\end{equation}
Notice that the opponent's reach probability appears as the
normalization factor in the computation of the action-utilities,
canceling the weighting factor and eliminating the need to compute
 these quantities unless one actually wishes to compute the
utility.  At the same time, this removes the possibility of division by zero should $P^{-i}(I)=0$. The strategy at the next iteration is proportional to the
amount of positive  regret
\begin{equation}
  b_{n+1}(I,a) = \left\{
  \begin{array}{ll}
    \frac{\max(R_n(I,a),0)}{\sum\limits_{a\in A(I)} \max(R_n(I,a),0)}, & \mbox{if}
    \sum\limits_{a\in A(I)} \max(R_n(I,a),0) > 0, \\
    \frac{1}{|A(I)|}, & \mbox{otherwise}.
  \end{array}
  \right.
\end{equation}
Finally, it is the  average of the sequence of strategies $b_n(I,a)$, weighted
by the reach probability of the player who controls $I$,
that converges to a NE:
\begin{eqnarray}
  \bar{b}_n(I,a) &=& \frac{\sum_{k=1}^n P^i(I;b_k) b_k(I,a)}{\sum_{k=1}^n P^i(I;b_k)}, \\
    P^{i}(I)&=& \prod_{j=1,I(a_j^x)\in {\cal I}^i}^{J_x} b(I(a_j^x),a_j^x), \, \forall x\in I.
\end{eqnarray}

\subsection{Extensive-Form Fictitious Play}
The key advantage of CFR  over (normal form) FP is the
ability to focus on one information set at a time. As demonstrated by
Heinrich et. al. \cite{HLS}, FP can be reformulated so that it too has this
feature. They do this by calculating a sequence of behavior strategies
that is realization equivalent to the sequence of normal form
strategies
\begin{equation}\label{GFP}
\sigma^i_{n+1}=(1-\alpha_{n+1}) \sigma^i_n + \alpha_{n+1} \beta^i(\sigma^{-i}_n),
\end{equation}
where $\beta^i$ is a  best response to the opponent's current strategy $\sigma^{-i}_n$.
This is a generalized FP with
weights $\alpha_n$ that decay to zero with a diverging sum $\sum
\alpha_n=\infty$, and reduces to the classic FP algorithm when
$\alpha_n=\frac{1}{n}$.  Their convergence proof relies on a result of
Leslie and Collins \cite{LC}, who define the following class of  generalized
  fictitious play algorithms, and then proceed to show that any such
algorithm converges to a NE of a zero-sum game.

\begin{definition} A {\em generalized weakened fictitious play process} is any process
$\{\sigma_n\}_{n\ge0}$, with $\sigma_n\in \Sigma$, such that
$$
\sigma_{n+1} \in \{(1-\alpha_{n+1})\sigma_n+\alpha_{n+1}(\beta_{\epsilon_n}(\sigma_n)+
M_{n+1})\}_{\beta_{\epsilon_n}},
$$
where $\beta_{\epsilon_n}=(\beta^1_{\epsilon_n},\beta^2_{\epsilon_n})$ is in the set of  $\epsilon_n$-best response vectors, $\alpha_n \rightarrow 0$, $\epsilon_n \rightarrow 0$ as $n \rightarrow \infty$,
$$
\sum_{n\ge 1} \alpha_n = \infty,
$$
and $\{M_n\}_{n\ge1}$ is a sequence of perturbations such that, for any $T>0$,
$$
\lim_{n \rightarrow \infty}
\sup_k \{ ||\sum_{j=n}^{k-1} \alpha_{j+1} M_{j+1}||
 :
\sum_{j=1}^{k-1} \alpha_{j+1} \le T \}=0.
$$
\end{definition}

 Note that in (\ref{GFP}), $\epsilon=0$ and $\beta^i$ is a best
 response.  Neither XFP nor GXFP make use of $\epsilon$-best
 responses, and we will later use the $\epsilon$ subscript to indicate
 a strategy in a perturbed game where actions must be taken with
 finite probability.

The L\&C \cite{LC} theorem relies on a result of Bena\"{i}m et. al. \cite{BHS}, which we will also need in the following section. We present the theorem in the form given by L\&C.

\begin{theorem}[Bena\"{i}m, et. al.]
  Assume $F: \mathbb{R}^m \rightarrow \mathbb{R}^m$ is a closed set-valued map
  such that $F(\sigma)$ is a non-empty compact convex subset of $\mathbb{R}^m$ with
    $$
\sup{||z|| : z \in F(z)} \le c(1+||\sigma||) \quad \forall \sigma.
$$
Let $\{\sigma_n\}_{n\ge0}$ be the process satisfying 
$$
\sigma_{n+1} -\sigma_n -\alpha_{n+1} M_{n+1} \in \alpha_{n+1} F(\sigma_n),
$$
with $\alpha_n \rightarrow 0$ as $n \rightarrow \infty$,
$$
\sum_{n\ge 1} \alpha_n = \infty,
$$
and $\{M_n\}_{n\ge1}$ be a sequence of perturbations such that, for any $T>0$,
$$
\lim_{n \rightarrow \infty}
\sup_k \{ ||\sum_{j=n}^{k-1}
\alpha_{j+1} M_{j+1}||
 :
\sum_{j=1}^{k-1} \alpha_{j+1} \le T \}=0.
$$
The set of limit points of $\{\sigma_n\}$ is a connected internally chain-recurrent set of the differential inclusion
\begin{equation} \label{BRDI}
\frac{d}{dt} \sigma_n \in F(\sigma_n). 
\end{equation}

\end{theorem}

L\&C first show that any GFP (\ref{GFP}) satisfies the requirements of
this theorem, and then show that the set of limit points is the set of
NE.

\begin{theorem}[Leslie and Collins] Any generalized weakened fictitious play process will converge to the set of NE in two-player zero-sum games, potential games, and generic 2 $\times $ 2 games.
\end{theorem}

Finally, Heinrich et. al. show that the mapping from the normal form strategies (\ref{GFP}) to behavior
form strategies requires
\begin{equation}\label{XFP}
b_{n+1}(I,a)=b_n(I,a) +
\frac{\alpha_{n+1}P^i(I;B^i_{n+1}) (B^i(I,a;b^{-i}_n)-b_n(I,a))}
{(1-\alpha_{n+1})P^i(I;b_n^i)+\alpha_{n+1} P^i(I;B^i)},
\end{equation}
where $B^i$ is a best response behavior strategy to $b_n^{-i}$:
\begin{equation} \label{BR}
B_{n+1}^i \in \arg \max_{b^i} u^i(b^i,b_n^{-i}),
\end{equation}
and $P^i$ is player $i$'s reach probability for either the current
strategy or the current best response.

When implemented using only (\ref{XFP}-\ref{BR}), we found XFP failed
to maintain normalized behavior strategies due to the accumulation of
round-off error.  In some cases this error became so significant that
the total exploitability could not be reliably calculated, sometimes
giving a negative result.  For this reason, we implemented XFP with a
renormalization step at each iteration:
\begin{equation}
\hat{b}_{n+1}(I,a)=b_{n+1}(I,a)/\sum_a b_{n+1}(I,a),
\end{equation}
using $\hat{b}_n$ in place of $b_n$ in (\ref{XFP}).

\subsection{Generalized Extensive Form Fictitious Play}

While GXFP (\ref{GXFP}-\ref{BD}) can be implemented with best responses
(\ref{BR}) rather than best decisions (\ref{BD}), we found it to be
significantly more accurate with the latter, in that the residual
exploitability due to round-off error was much smaller. We found the
opposite to be true
for XFP---it converged with either type of what we will call {\em better responses}, but was significantly more accurate with best responses.  The
convergence proof for GXFP is essentially the same in both cases, with
either choice requiring an expansion of Definition 1 to include a
 weakened notion of better response that will be described in the
 next section.

In practice (\ref{BD})  can be computed by simply choosing any optimal action:
$$
d(I,a;b_n) = \delta_{a\tilde{a}},\mbox{ where } \tilde{a} \in \arg \max_{a \in I} U(I,a;b_n),
$$ where $\delta_{a\tilde{a}}$ is the Kronecker delta. This is a much
more straigt-forward procedure than computing a best response $B^i$, which
requires navigating the game tree from the leaves up toward the root.

Finally, GFXP can be implemented with a more general weight
$\alpha_n$, but we chose to use the simple average of best decisions
that can be computed by {\em counting} the number of times each action
is best at a given information set, thus avoiding the accumulation of
round off error. This option, which is not available for CFR and XFP,
may be useful for extremely large games.

\section{Convergence of Generalized Extensive-Form Fictitious Play}

To prove that GXFP converges, we first expand Definition 1 to include
alternative ``better'' responses. We then show that GXFP is equivalent
to one of these expanded GFPs.  Next, we use Theorem 1 to adapt
Theorem 2 to establish convergence. Finally we adapt a theorem due to
Hofbauer \& Sorin \cite{HS} to prove  that the attractive set is the set of NE.

We will need the mapping from  behavior strategies to
   mixed strategies:
\begin{equation}\label{map}
  \sigma^i(s;b)=\prod_{I \in {\cal I}^i} b(I,a(s,I)),
\end{equation}
where where $a(s,I)$ is the action $a\in A(I)$ consistent with the pure strategy $s\in S^i$.
The  better response  required to show that (\ref{GXFP}) converges  takes the form
\begin{equation}
\tilde{d}^i(s;b)=\frac{1}{|{\cal I}^i|}\sum_{\bar{I} \in {\cal I}^i}  d(\bar{I},a(s,\bar{I});b) \prod_{I \neq \bar{I},I\in {\cal I}^i} b(I,a(s,I)),
  \end{equation}
where $d$ is the best decision (\ref{BD}) introduced earlier. Note that $\tilde{d}^i$ is a mixed strategy, whereas $d$ is a behavior strategy, and that 
 these  depend on the strategy of both opponents.
In view of (\ref{map}), we  will use  $\tilde{d}^i(s;b)\equiv
\tilde{d}^i(\sigma)$ and $\tilde{d}(\sigma)$ to indicate the  better response vector,  as we have done with best responses.
The proofs given below hold  with
best responses (\ref{BR}) replacing best decisions to define  $\tilde{B}^i(s;b_n)$, but, as noted earlier, we found this to be computationally less accurate.

\begin{theorem} GXFP is realization equivalent to a generalized weakened FP with best responses $\beta_\epsilon(\sigma)$ replaced by weakened better decisions $\tilde{d}$.
\end{theorem}

\begin{proof}

  Inserting (\ref{BD}) into the mapping from  behavior strategies to
   mixed strategies gives
\begin{eqnarray}
  \sigma_{n+1}^i(s;b_{n+1})&=&\prod_{I \in {\cal I}^i} b_{n+1}(I,a(s,I)) \nonumber \\
  &=& \prod_{I \in {\cal I}^i} \left[\left(1-\frac{1}{n+1}\right) b_n(I,a(s,I)) + \frac{1}{n+1} d(I,a(s,I);b_n)\right]. \nonumber
\end{eqnarray}

Next, we isolate terms of $O(\frac{1}{n})$ and larger  from  the product
\begin{eqnarray}
  \sigma_{n+1}^i(s;b_{n+1})&=&
  \left(1- \frac{|{\cal I}^i|}{n+1}\right)\prod_{I \in {\cal I}^i} b_n(I,a(s,I))\nonumber \\
  & & + \frac{1}{|{\cal I}^i|}\sum_{\bar{I}
    \in {\cal I}^i} \frac{|{\cal I}^i|}{n+1} d(\bar{I},a(s,\bar{I});b_n) \prod_{I \neq \bar{I},I\in {\cal I}^i} b_n(I,a(s,I)) \nonumber \\
  & & +\frac{|{\cal I}^i|}{n+1}M^i_{n+1}(s;b_n),\label{this}
\end{eqnarray}
where we have removed a factor of $\frac{|{\cal I}^i|}{n+1}$ from the
finite number of higher order terms and grouped what remains into the perturbation
$M^i_{n+1}$. Letting
$$
\alpha_n = \frac{|{\cal I}^i|}{n},
$$
and rearranging (\ref{this}) gives
$$
\sigma_{n+1}^i \in \{(1-\alpha_{n+1})\sigma_n^i+\alpha_{n+1}(\tilde{d}^i+
M^i_{n+1})\}_{\tilde{d^i}}.
$$
The weights $\alpha_n$ are proportional to those in standard FP and satisfy the requirements in Definition 1, while the perturbations $M_n=(M^1_n,M^2_n)=O\left(\frac{1}{n}\right)$ decay sufficiently fast to ensure the requirement
$$
\lim_{n \rightarrow \infty}
\sup_k \{ ||\sum_{j=n}^{k-1}
\alpha_{j+1} M_{j+1}||
 :
\sum_{j=1}^{k-1} \alpha_{j+1} \le T \}=0. 
$$

\end{proof}

In the L\&C result \cite{LC}, the relevant differential inclusion is
(\ref{BRDI}).
In view of Theorem 3, we must consider instead the weakened better
response defined above.

\begin{theorem}
  The set of limit points of a generalized weakened fictitious play process with best responses $\beta_\epsilon(\sigma)$ replaced by weakened better decisions $\tilde{d}$
   is a connected internally chain-recurrent set of the differential inclusion
\begin{equation}\label{DI}
\frac{d}{dt} \sigma_n \in \{ \tilde{d}(\sigma_n) - \sigma_n\}_{\tilde{d}} \equiv F(\sigma_n).
\end{equation}
\end{theorem}

\begin{proof}
 $F(\sigma)$
 satisfies the requirement of Theorem 1:
  $$
\sup{||z|| : z \in F(z)} \le c(1+||\sigma||) \quad \forall \sigma.
$$
The requirements on the perturbation $M_n$ in Definition 1 are the same as in Theorem 1, and were already shown to be satisfied in the proof of Theorem 3.
\end{proof}

Finally, we must show the attractive set is the set of NE:
\begin{theorem}
Any generalized weakened fictitious play process with  best responses $\beta_\epsilon(\sigma)$ replaced by weakened better decisions $\tilde{d}$  will converge to the set of NE in two-player zero-sum games.
\end{theorem}

\begin{proof}
  We must show that the attractors of the differential inclusion
  (\ref{DI}) are the set of NE. We will follow the proof of Hofbauer
  and Sorin \cite{HS}, who show this for (\ref{BRDI}). They do this by
  considering the total exploitability,
  \begin{eqnarray} \label{TE}
    v(t)&=&V(\sigma^1(t),\sigma^2(t))=\max_{\sigma^1} u^1(\sigma^1,\sigma^2) - \min_{\sigma^2} u^1(\sigma^1,\sigma^2) \nonumber \\
    &=& \sum_{i\in\{1,2\}} u^i(\beta^i(\sigma^{-i}),\sigma^{-i}) \ge 0.
    \end{eqnarray}
They show that $v(t)$  evolving under the best response differential inclusion (\ref{BRDI})
  satisfies
  $$
\frac{d}{dt} v(t) \le -v(t),
$$ implying
$$
v(t) \le e^{-t} v(0),
$$ so that $v(t)$ decays to zero. They also show how to adapt this to the discrete dynamics of a FP process. An unexploitable
strategy pair is a NE by definition.

The proof given by Hofbauer \& Sorin applies to (\ref{DI}) if we replace the best responses with the weakened better response $\tilde{d}^i$: 
$$ \tilde{v}(t)=\widetilde{V}(\sigma(t))=\sum_{i\in\{1,2\}}
u^i(\tilde{d}^i(\sigma),\sigma^{-i}), 
$$
as the result ultimately follows from the convexity of $V$, which also holds for $\widetilde{V}$.
Thus,  their arguments  translate to
  $$
\frac{d}{dt} \tilde{v}(t) \le -\tilde{v}(t),
$$
implying
$$
\tilde{v}(t) \le e^{-t} \tilde{v}(0),
$$
with $\tilde{v}(t)$ decaying to zero. When this alternative measure of exploitability $\tilde{v}(t)$ reaches zero there will be no information set at which either player can unilaterally improve. Working backward through the game tree, this implies that both players are playing a best response, hence we are at a NE.

\end{proof}

\section{Benchmark Game}

The analysis of large games is often absent from  books on
game theory, which tend to focus on extremely simple games,
e.g. rock-paper-scissors or the Prisoner's Dilemma.  A notable exception to
this occurs in {\em The theory of games and economics behavior} by von
Neumann and Morgenstern \cite{NM}, which features an in depth analysis of two
models for the game of poker. The text refers to these as the
symmetric and asymmetric games.  A footnote at the opening of this
discussion reveals that these models were largely responsible for von
Neumann's original exploration of game theory.
Poker itself is so complicated that few take on the task of using
something like Texas Hold'em as a benchmark, so reduced models are
often used for the purpose of exploring RL methods. Many of these
models can be solved exactly.

Of the two models considered by vN\&M, the asymmetric one is more
similar to the actual game of poker. Further, it is more challenging
from a computational perspective, as it possesses an infinite number
of NE.  While this is a useful starting point for benchmarking, the
game tree is not deep enough for the play-to-reach feature to matter,
as neither player has any control over which of their own informaton
sets are reached.  The game of poker, however, suggests many ways of
generalizing this game into a broad family of potential benchmark
games that are still relatively easy to implement.

In this section, we first review the results from vN\&M
 for their asymmetric game, and then introduce a
generalized version that features a somewhat deeper game tree.

\subsection{Von Neumann and Morgenstern's asymmetric game}

The solutions to large games can be bewildering from a human
perspective. A further advantage of simplified models is that  one
can gain some intuitive understanding of how more complicated games
work. Starting with vN\&M's asymmetric game will help us understand
what is happening in the benchmark game we consider in the next section.

In these two-player zero-sum games, players are `dealt' hands
(private information) that take the form of random numbers.  The
discrete version of these games, where the players are dealt random
integers, $1 \le i \le N$, is mentioned in vN\&M, but the text
principally focuses on the continuous versions of these games, with hands $x 
\in [0,1]$. The continuous versions 
 are more readily  solved exactly, and the solutions for the
asymmetric game closely track that of the discrete game for large $N$.
This will be useful for understanding our numerical solutions.

In the asymmetric game\footnote{We have adopted a more modern poker
  parlance, but the game is equivalent to the version described by von
  Neumann and Morgernstern with the ``low bid'' and ``high bid''
  options.}, each of the players {\em ante} an amount $A$, forming a {\em pot} $P=2A$, and then take
turns deciding whether or not to place an additional bet $B$.  The
players are betting on their private information---the value of a
random number `dealt' to them after placing their antes, but before
making the additional bets.  The first player to act may either {\em
  check} (i.e. bet zero) or place a {\em bet} of fixed size $B>0$. In
the vN\&M model, the second player only acts if this bet is placed,
and then has the option to {\em fold} or {\em call}. If the second
player folds, the first player receives the antes. If the fist player
checks or the second player calls, the pot is distributed according to
the highest hand.  In their text, vN\&M 
 introduce what would later
be called a behavior-strategy description of these games, where
$b^i(x,a)$ describes the fraction of time player $i$ takes action $a$
at each information set.  They find this game to have a continuum of
optimal solutions (equivalent to NE, which had not yet been
invented). The first player plays the same strategy in all of these,
having two thresholds between which they never bet, and outside of
which they always bet:
\begin{eqnarray}
  x_1&=&\frac{AB}{4A^2+5AB+B^2}, \nonumber \\
    x_2 &=& \frac{2A^2+4AB+B^2}{4A^2+5AB+B^2} \nonumber .
\end{eqnarray}
The lower region corresponds to a {\em bluff}.  In the modern
recreational poker literature, betting one's weakest and strongest
hands is referred to as betting a {\em polarized} range. Despite
its simplicity, the  model captures this significant
insight into poker strategy.  The second player has an
infinite number of choices that achieve NE.
If we let
$b^2(y,\mbox{call})$ be the fraction of time the second player calls a bet
when facing one, vN\&M show that
$$
\frac{1}{x_2-z_0} \int_{z_0}^{x_2} b^2(y,\mbox{call}) dy   \left\{
\begin{array}{ll}
  =  \frac{A}{A+B} &\mbox{if }\, z_0=x_1       \\
 \ge \frac{A}{A+B}  &\mbox{if }\, x_1<z_0<x_2      
\end{array}
\right.
$$
are
both necessary and sufficient conditions for equilibrium. Among these choices, there is a single, weakly-dominant strategy where the second player always folds/calls below/above a threshold
$$
y_1= \frac{3AB+2B^2}{4A^2+5AB+B^2} \nonumber .
$$

\begin{figure}
  \begin{center}
    \includegraphics[height=3.in]{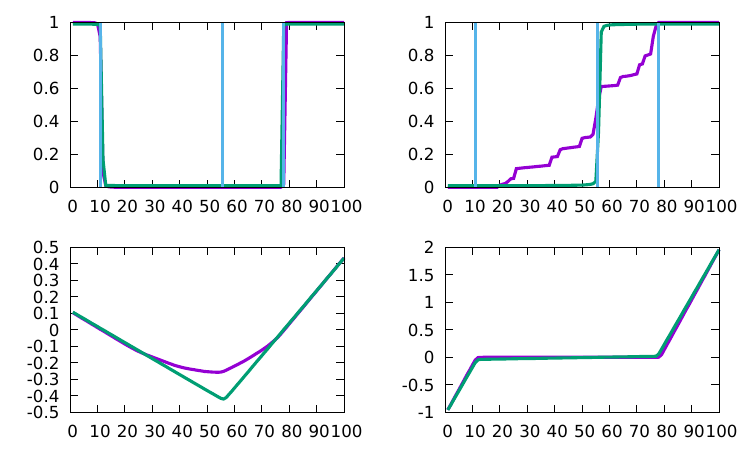}
    \end{center}
  \caption{The discrete version of vN\&M's asymmetric game with $P=1, B=1$ and
    $N=100$ hands.  The purple curves correspond to no perturbation, $\epsilon=0$; the
    green curves are the perturbed game with 
$\epsilon=0.01$.
    Top row: the fractions with which player 1 bets
    (left panel) and player 2 calls (right panel), along with the thresholds for the continuous game (blue vertical lines). Bottom row:
    the difference in the expected
value of player 1's options (left panel) and player 2's options (right
panel).}
\end{figure}

Numerical solutions using any of the algorithms described above reveal
an analogous result, where player 2's strategy endlessly drifts among
weakly dominated strategies.  For this reason, we consider both the
perturbed game, where each option at a given information set is played
with a minimum probability, $b(I,a) \ge \epsilon$, and the
non-perturbed game.  This is straightforward with GXFP, as we simply
constrain the best decision:
\begin{equation} \label{CBD}
  d_\epsilon(I,a) = \epsilon+ (1-\epsilon) \delta_{a\tilde{a}},\mbox{ where } \tilde{a} \in \arg \max_{a \in I} U(I,a;b_n).
\end{equation}
A similar calculation can be made when computing best responses with XFP, while
CFR requires a somewhat more complicated adjustement \cite{FKS}. In Figure 1, we
present results for the game with $N=100$ uniformly distributed
hands. These results are analogous to those shown vN\&M.

The graphs in the top row are the fractions with which player 1 bets
(left panel) and player 2 calls (right panel). For player 1, the
unperturbed result (shown in purple) and the perturbed result with
$\epsilon=0.01$ (shown in green) are nearly indistinguishable,
indicative of there being a unique strategy for player 1 at
equilibrium.  For comparison, the thresholds given above for the
continuous version of the game are shown as blue vertical lines. For
player 2, the perturbed result is unique and approximates the
solution with pure-strategy thresholds given above, while the unperturbed result
endlessly drifts among the set of NE that employ a weakly-dominated strategy for player 2.

In the bottom row of Figure 1, we graph the difference in the expected
value of player 1's options (left panel) and player 2's options (right
panel), with a positive difference corresponding to the betting and
calling options, respectively. From the left panel, we see that player
2's weakly dominated strategy is outperformed by the dominant one if
player 1 is forced to bet with probability $\epsilon$ in the region
where the unperturbed strategy is to check. Examining the lower-right
panel, we see that there is no significant difference in player 2's
perturbed and unperturbed payoff, as player 1 plays a nearly equal
strategy in each case. We also see that the ambivalence in player 2's
strategy is due to being indifferent between calling and folding
throughout the region where mixed equilibria exist.

\subsection{The benchmark game}

If one understands {\em fictitious play} to mean finding the average
best response to an opponent's prior play, and does not realize this
was originally only considered for the normal form of a game, it is
easy to stumble upon  algorithm (\ref{GXFP}-\ref{BD}). The situation
is more subtle, however, if the game tree is deeper. This requires the
play-to-reach assumption.  This issue is irrelevant in the case of the
game just discussed. To bring it to the forefront, we will consider a
slightly more complicated version of the asymmetric game that allows
for a {\em bet} and a single {\em raise}, including the possibility of
a {\em check-raise}. This game is briefly addressed in  Chen \&
Ankenman \cite{CA}, a book aimed primarily at recreational poker
players, but the discussion omits the details given below.

For the continuous version of the game, a NE using pure-action choices
containing 12 thresholds exists, but the strategy for player 1 is
weakly dominated by an infinite number of other stategies, including
pure-action strategies that features two additional thresholds. This
differs from what happens in the vN\&M game, where there is a unique
pure-action NE with the smallest number of thresholds and weakly
dominant strategies.  For the expanded game just described, the linear
system of equations that would determine the full set of thresholds is
singular, leading to a degeneracy of pure action equilibria. The
numerical results presented in the next section reveal that this game
also features an infinite number of mixed strategy NE, with 
 this occurring for both players.

When the pot $P=2A=1$, bet  $B=1$ and raise  $R=1$,
the eight thresholds for player 1 are $\{ x_1 = 64/1083, x_2 =
369/722, x_3 = 10/19, x_4 = x_5-32/1083, x_5 , x_6 = 307/361, x_7 =
x_8- 22/361,x_8 \}$, where $x_5$ and $x_8$ can be chosen arbitrarily
so long as all of the thresholds remain in ascending order. These
correspond to nine intervals of hands where player 1 takes a specific
sequence of actions: $\{ \mbox{bet-fold} < \mbox{check-fold} <
\mbox{check-raise} < \mbox{check-call} < \mbox{bet-fold} <
\mbox{check-call} < \mbox{bet-call} < \mbox{check-raise} <
\mbox{bet-call} \}$.  The six thresholds for the second player divide
into two sets of three: $\{y_1^1=8/57, y_2^1= 41/57, y_3^1=15/19 \}$,
corresponding to four intervals where player 2 responds to a check with
the actions
$\{\mbox{bet-fold}<\mbox{check}<\mbox{bet-fold}<\mbox{bet-call}\}$ and
$\{ y_1^2=1/2,y_2^2=10/19, y_3^2=17/19\}$ where player 2 responds to a
bet with the actions
$\{\mbox{fold}<\mbox{raise}<\mbox{fold}<\mbox{call}\}$.

\section{Numerical experiments}

As with the asymmetric game, the degeneracy of the solution means that
the numerical solution will drift among the possible equilibria. This
makes direct comparison of the strategy profiles generated by the
methods difficult. Thus in our first figure in this section we present
only a single realization generated by GXFP.  The CFR and XFP algorithms give
qualitatively similar results, as do other initial conditions.  In the
rest of the section we will be able to make direct comparisons.

\begin{figure}
  \begin{center}
    \includegraphics[height=3.in]{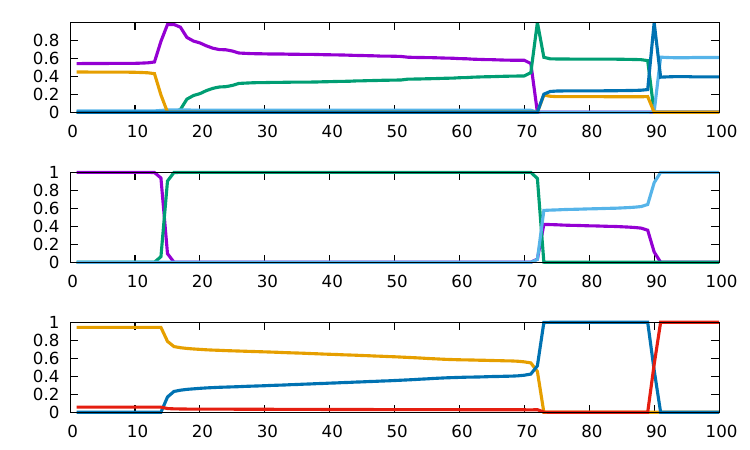}
    \end{center}
 \caption{Numerical solution of the  benchmark game with $P=1, B=1, R=1$ and $N=100$
   hands. Top panel: the probability with which player one takes the
   action sequences check-fold (purple), check-call (green), check-raise
   (light blue), bet-call (dark blue), and bet-fold (gold).  Middle
   panel: player 2's response to an initial check from player 1:
   either check (green), bet-fold (purple), or bet-call (light blue). Bottom
   panel: player 2's response to an initial bet from player 1: either
   fold (gold), call (dark blue), or raise (red).}
\end{figure}

In Figure 2 we plot the strategies returned by the GXFP algorithm as a
function of the hand strength, $1 \le i \le 100$. In the top panel,
the purple curve is the probability of checking, followed by folding
to a bet; the green curve is the probability of checking, followed by
calling a bet, and the light blue curve is the probability of
checking, followed by raising, these three quantities adding to
one. The remaining two strategy sequences, bet-call (dark blue) and bet-fold (gold)  also add to one. 
The middle panel is player 2's response to
an initial check from player 1: either another check (green), a bet
followed by a fold if raised (purple), or a bet followed by a call if
raised (light blue), these three quantities adding to one. The bottom panel
is player 2's response to an initial bet from player 1: either  fold
(gold),  call (red), or a raise (dark blue), with these three
quantities again adding to one. 
The main thing to notice here is that there is a
mixture of strategies being used for most hands, indicative of the type
of degeneracy we saw in the asymmetric game, only this time it 
occurs for both players.

\begin{figure}
  \begin{center}
    \includegraphics[height=3.in]{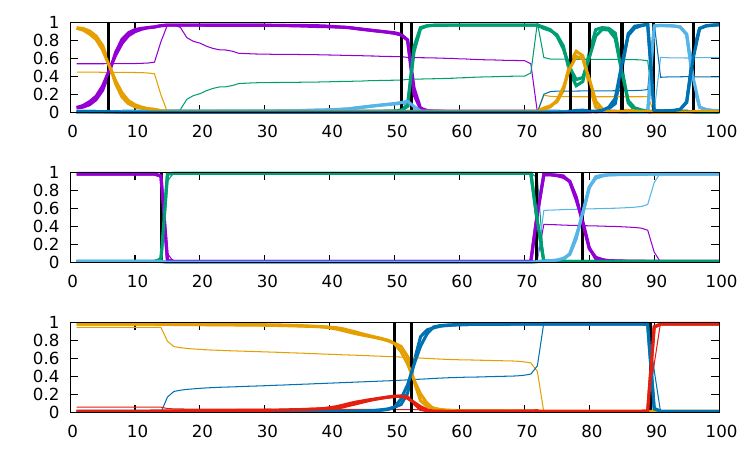}  
    \end{center}
  \caption{Numerical solution of the perturbed
    ($\epsilon=0.01$) benchmark game using XFP and GXFP (nearly identical; both using the color scheme described in the last figure) and the unperturbed ($\epsilon=0.0$) benchmark game  using GXFP (all blue, thin line) for $P=1, B=1, R=1$, and $N=100$, along with
  the thresholds (black vertical lines) for the continuous version of the game.}
\end{figure}

As with the asymmetric game, we find that we can again remove the
weakly dominated strategies by approximating the equilibria subject to
a small perturbation. 
The results for doing this with XFP and GXFP are
shown in Figure 3, along with the unperturbed solution from Figure 2
and the thresholds for the continuous version of the game. The two
arbitrary thresholds ($x_5$ and $x_8$) were roughly fit to these
graphs, but the fit of all other thresholds provides a useful way of
detecting coding errors.  The mixing near the thresholds is due to
boundary effects inherent to the discrete game, and diminishes as the
number of possible hands increases. Some of these regions are very
narrow, so the discretization affects the result more strongly.  The
results for the two algorithms are nearly identical, as expected. We
did not adapt CFR to play the perturbed game.

To make comparisons between all three algorithms, we examine plots of
the utility of the current strategy pair $u^1(b_n^1,b_n^2)$ measured from
player 1's perspective and the total exploitability $v(t) \ge 0$, defined in equation (\ref{TE}). In principal, the
former should converge to the value of the game, but in practice there
is some numerical error due to finite precision arithmetic. The total
exploitability serves as a measure of this error.

\begin{figure}
  \begin{center}
    \includegraphics[width=6 in,height=6 in]{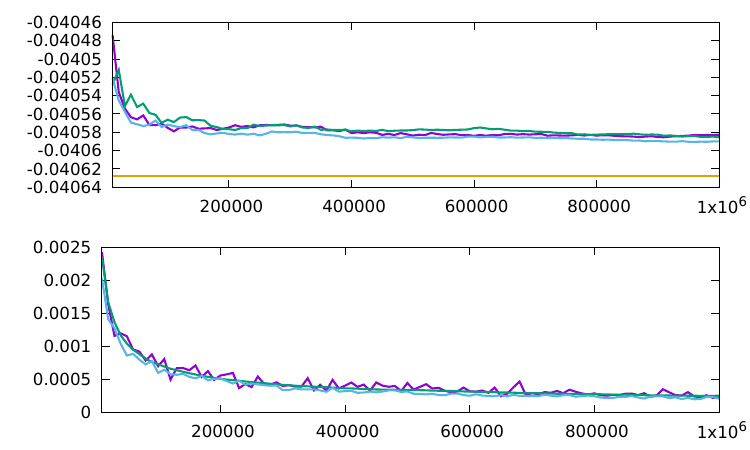} 
    \end{center}
  \caption{Numerical solution of the benchmark game with $P=1, B=1, R=1$, and $N=100$ hands. Top panel: the value of the game (from player 1's
perspective) returned by the three algorithms at intervals of
$10,000$ iterations, along with the  exact value for the continuous version of the game,  
$-\frac{44}{1083} \approx 0.406$, shown in gold. A sample for each of the 
of the  algorithms is shown for random initial data:  GXFP 
in purple, CFR in light blue and XFP in green. Bottom panel: the corresponding plots for total exploitability.}
\end{figure}

In the top panel of Figure 4 we plot the expected value of the current
strategy pair (from player 1's perspective) returned by the three
algorithms at intervals of $10,000$ iterations, along with the the
exact value of $-\frac{44}{1083} \approx 0.406$ for the continuous
version of the game, shown in gold.  A sample for each of the three
algorithms is shown for random initial data. The solutions oscillate
somewhat at the scale shown, but are all comparable in
magnitude. Results vary with initial conditions, but are qualitatively
similar. The numerical solutions will get closer to the solution for
the continuous version of the game as $N \rightarrow \infty$, but
round off error will prevent them from achieving this solution
exactly.

While getting the correct value of the game is desirable, this can
sometimes be achieved with nonoptimal strategies.  A better measure of
error is the total exploitability. In the lower panel, we plot this
quantity at every $10,000$ iterations for the same samples presented
in the top panel. The error and the rate at which it decays is similar
for all three methods.

Finally, we should note that all three
algorithms have a similar computational cost per iteration, despite
the somewhat more complicated  XFP and CFR
recursive updates. The extra calculations needed to compute regret or a best response scale with the number of information sets, $|{\cal I}^1| + |{\cal I}^2|$, but the most expensive part of these calculations, computing  the action utilities, must be done for all three algorithms, and this scale
like $|{\cal I}^1|  |{\cal I}^2|$.

\begin{figure}
  \begin{center}
    \includegraphics[width=6 in,height=6.in]{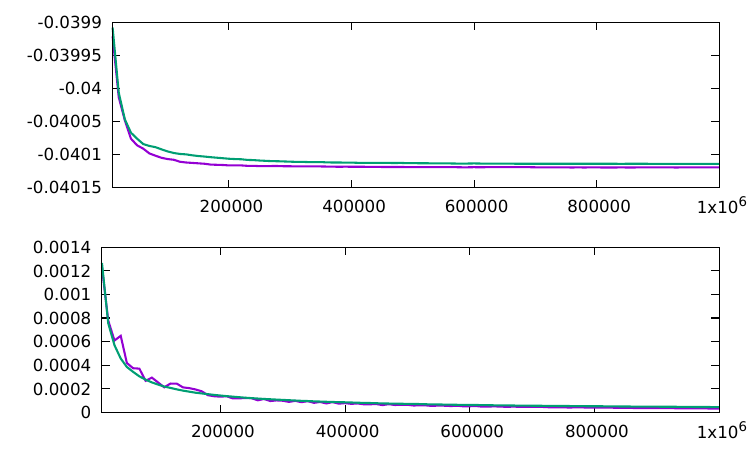} 
    \end{center}
  \caption{The perturbed ($\epsilon=0.01$) benchmark game with $P=1,
    B=1, R=1$, and $N=100$ hands. Top panel: the value of the game
    (from player 1's perspective) returned at intervals of $10,000$
    iterations.  A sample of each algorithm is shown for random
    initial data: GXFP in purple and XFP in green. Bottom panel: the
    corresponding plots for total exploitability.}

\end{figure}

In Figure 5, we compare the performance of XFP and GXFP for
generalized NE with $\epsilon=0.01$. 
In the top panel, the purple and green curves are once
again the utility of the current strategy pair from player 1's
perspective returned by the GXFP and XFP algorithms at intervals of
$10,000$ iterations. The bottom panel is the exploitability using the
same color scheme. Note that convergence for the perturbed game is
much faster than for the unperturbed one, a result of the degenerate
NE making convergence much more difficult to achieve.

\section{Summary}

In this work we have introduced a new algorithm, GXFP, that is
realization equivalent to a generalized form of Fictitious Play, thus
inheriting the convergence properties of that class of algorithms. We
then compared the computational performance of GXFP to that of two
additional algorithms, CFR and XFP, using an expanded version of
a simple poker model first introduced by von Neumann and
Morgenstern. This benchmark game, where each of the two players makes at least
one, and up to two, decisions, has a somewhat deeper game tree that
requires the play-to-reach utilities that are common to all three
algorithms.  We also presented an exact solution for the continuous
version of this game that is useful for testing the algorithms. Like
vN\&M's original game, this game features an infinite number of
NE.  As a result, a variation of the game with a perturbed
strategy space and a unique equilibrium was also considered.  This
generalized NE is computed more quickly and is easier to interpret.

The computational cost per iteration is comparable for all three
algorithms.
GXFP's simple update formula requires no best response
calculation, relying instead on a best decision calculation that is
intuitive and, at least in some approximate sense, routinely used 
to make decisions in recreational games. This makes it an ideal choice for
anyone looking for a quick and easy game solving tool.  Unlike the
other two algorithms, GXFP can be implemented without accumulating
round-off by using integer type variables to count the number of times a
strategy choice was best. This can be done with or without the
perturbed strategy space.  GXFP also has somewhat lower memory
requirements. These features may be helpful for solving especially
large games.

Finally, all three algorithms converged much more quickly when using
updates that alternate between the two players, using the opponent's
most recently updated strategy rather than the strategy from the
previous iteration. This is a well-known feature of algorithms of this
type, and is, for example, largely responsible for CFR+'s faster
convergence \cite{T,BMS}.

\vskip 0.2in
\bibliography{paper37}
\bibliographystyle{theapa}

\end{document}